\documentclass[a4paper,11pt]{article}
\usepackage[utf8]{inputenc}
\usepackage{CommandsAndStyle}

\usepackage[colorinlistoftodos,bordercolor=orange,backgroundcolor=orange!20,linecolor=orange,textsize=scriptsize]{todonotes}

\usepackage[ruled, vlined, norelsize, noresetcount]{algorithm2e}

\SetAlgorithmName{Algorithm RSHT}{rsht}{rsht}

\usepackage{subcaption}
\usepackage{booktabs}
\usepackage{paralist}
\usepackage{fixltx2e}
\usepackage{pgf,tikz}
\usepackage{mathrsfs}
\usetikzlibrary{arrows}

\newtheorem{question}{Question}[section]
\newtheorem{conjecture}[question]{Conjecture}

\title{Random Simple-Homotopy Theory}
\author{Bruno Benedetti\thanks{Department of Mathematics, U Miami, Coral Gables, FL 33146, USA; \texttt{bruno@math.miami.edu.}}, 
            
            Crystal Lai\thanks{Institut f\"ur Mathematik, TU Berlin, Stra{\ss}e des 17.\ Juni 136, 10623 Berlin, Germany; \texttt{$\{$lofano, lutz$\}$@math.tu-berlin.de}.}\,\,\,\,, 
            Davide Lofano$^{**}$, Frank H.~Lutz$^{**}$}
\date{September 25, 2021}

\begin{document}

\maketitle

\begin{abstract}
We implement an algorithm \emph{RSHT (Random Simple-Homotopy)} to study the simple-homotopy types of simplicial complexes, 
with a particular focus on contractible spaces and on finding substructures in higher-dimensional complexes. 
The algorithm combines elementary simplicial collapses with \emph{pure elementary expansions}. 
For triangulated $d$-manifolds with $d\leq 6$, we show that RSHT reduces to (random) bistellar flips.
 
Among the many examples on which we test RSHT, we describe an explicit $15$-vertex triangulation of the Abalone, and more generally, 
$(14k+1)$-vertex triangulations of Bing's houses with $k$ rooms, $k\geq 3$, which all can be deformed to a point using only six pure elementary expansions.
\end{abstract}

\section{Introduction}

 \enlargethispage{4mm}
A standard task in topology is to simplify a given presentation of a topological space. In general, this task cannot be performed algorithmically: 
Even the simplest homotopic property, contractibility, is undecidable. Nevertheless, here we propose a simple randomized algorithm 
to modify triangulations while keeping the simple-homotopy type of a space. The algorithm can be used as a heuristic
to study simply-connected complexes, or, more generally, complexes whose fundamental group has no Whitehead torsion.
We shall see that in several contractible examples the heuristics works very well.  The algorithm is also of interest when applied to manifolds or complexes of arbitrary topology, as we discuss below. 

Our work builds on that of Whitehead, who in 1939 introduced a discrete version of homotopy theory, called \emph{simple-homotopy theory} \cite{whitehead1939simplicial}. 
An \emph{elementary collapse} is a deletion from a simplicial complex of a free face, i.e., of a non-empty face that is properly contained in only one other face,
along with that face it is contained in. Elementary collapses are deformation retracts, and thus maintain the homotopy type; the same is true for their inverse moves, 
\emph{elementary anticollapses}. Two simplicial complexes are of the same \emph{simple-homotopy type} if they can be transformed
into one another via some sequence of collapses and anticollapses, called a \emph{formal deformation} \cite{whitehead1939simplicial}.

Equivalently, two simplicial complexes are of the same simple-homotopy type if there exists a third complex that can be reduced 
to both the original ones via suitable sequences of elementary collapses \cite[p.~13]{Metzler1993}.
The size of the third complex (or, using the former definition, the length of the formal deformation) cannot be bounded a priori, 
because the simple-homotopy type cannot be decided algorithmically. In fact, by a famous result of Whitehead, having the simple-homotopy type 
of a point is equivalent to being contractible \cite{whitehead1939simplicial} and thus undecidable. 

In contrast, it is possible to decide algorithmically  whether a given complex is \emph{collapsible}, i.e., whether it can be reduced via collapses onto a single vertex. 
This decision problem was recently proved to be NP-complete by Tancer~\cite{tancer2016recognition}.
The advantage of the collapsibility notion is that all intermediate steps in the reduction are simplicial complexes of smaller and smaller size, hence very easy 
to encode and work with. The drawback is that collapsibility is strictly stronger than contractibility: 
Many ``elementary'' contractible complexes, like the Dunce Hat \cite{Zeeman64Dunce} or Bing's House with two rooms  \cite{Bing1964}, are not collapsible. 

In 1998, Forman introduced a second way to study contractibility combinatorially. His Discrete Morse theory \cite{forman1998morse,forman2002user} 
is a tool to reduce simplicial complexes using a mix of collapses and facet deletions. The advantage is that all simplicial complexes (contractible or not) 
can now be reduced to a vertex, possibly by using a relatively large number of facet deletions. 
The drawback is that even if one starts with a simplicial complex, the intermediate steps in the reduction sequence are typically non-regular CW complexes, 
and thus harder to handle. By only focusing on the count of facet deletions (the so-called ``discrete Morse vector'') it is possible to use randomness 
to produce fast implementations \cite{benedetti2014random}, but at the cost of failing to recognize many contractible complexes. 
See \cite{joswig2014}, \cite{ABL}, and \cite{lofano2019},  for computational and theoretical obstructions.

In this paper, we go back to Whitehead's original idea, and propose a third simplification method based on collapses in combination with certain expansions.
Our randomized heuristic \emph{Random Simple-Homotopy} (\emph{RSHT}; see Section~\ref{sec:implementation}) has two advantages: 
First, all intermediate steps are indeed simplicial complexes; and second, at the moment we do not know of a single contractible complex 
for which our heuristics has probability zero to succeed in recognizing contractibility.
 
Here is the idea. We perform elementary collapses until we get stuck. Then we select a top dimensional face $\varrho$ 
uniformly at random, and  for all $d$-faces $\varrho'$ adjacent to $\varrho$ via a $(d-1)$-dimensional ridge, we check if the subcomplex induced on the $d+2$ vertices 
of $\varrho \cup \varrho'$ is a \emph{pure $d$-dimensional ball}. This test is very fast. If for some $\varrho'$ the answer is positive, 
we glue onto our complex the full $(d+1)$-simplex $\sigma$ 
on the vertices of $\varrho \cup \varrho'$. If for several $\varrho'$s the answer is positive, 
we simply choose one uniformly at random. 

This glueing step is called a \emph{pure elementary $(d+1)$-expansion}, and it is also classical from the topological perspective, compare \cite[Chapter I]{Metzler1993}. 
After this step, we may collapse away the newly introduced $(d+1)$-simplex $\sigma$ together with any $d$-face $\tau$ of it. To avoid undoing the pure elementary expansion, 
we must select a~$\tau$ that was already present in the complex we got stuck at before the pure elementary expansion. This first elementary collapse 
after the pure elementary expansion  is called ``(CC) step'' below (see Section~\ref{sec:implementation}). The combination ``pure elementary expansion + (CC) step'', 
known in the topological literature as ``transient move'' \cite{Metzler1993}, maintains both the dimension and the simple-homotopy type: 
In fact, any pure elementary expansion can be viewed as a composition of back-to-back elementary anticollapses.
  
Whitehead proved that for every contractible complex there is a formal deformation that reduces it  to a single point \cite{whitehead1939simplicial}. 
It is not known if there is also a formal deformation to a point in which one performs anticollapses or expansions ``only when stuck'',
i.e., only to intermediate complexes without free faces. If this is true, then indeed any contractible complex would have a positive probability 
to be recognized by our heuristics. Of course, we cannot in any case expect any universal upper bound on the number of elementary anticollapses needed, 
or else we would have found an algorithm that recognizes contractibility. 

However, we shall see in Sections~\ref{sec:Dunce-Bing} and \ref{sec:substructures} that in many key examples the number of pure elementary expansions needed 
is relatively low. As a benchmark series, we build \emph{Bing's house with $k$ rooms}, a one-parameter generalization of the aforementioned Bing's house with two rooms.  
For all $k\geq 3$, we prove that Bing's house with $k$ rooms can be collapsed by adding only six further tetrahedra, cleverly chosen (\textbf{Theorem \ref{thm:Bingk}}). 
Of course, since our algorithm is randomized, there is no guarantee that precisely those tetrahedra will be selected as expansions. 
But even with a quick attempt consisting of $10^4$ runs, our algorithm is able to reduce Bing's house with seven rooms (which is a $2$-complex on $99$ vertices) 
to a point by adding only $41$ tetrahedra; see Table~\ref{tbl:examples}.

\emph{Random Simple-Homotopy (RSHT)} works with simplicial complexes of arbitrary dimension, but it is of particular interest 
when applied to triangulations of low-dimensional manifolds.  When $d \le 6$, we show (in \textbf{Theorem \ref{thm:bistellarflips}}) 
that on any (closed) $d$-manifold RSHT has basically the same effect of performing \emph{bistellar flips}, also known as \emph{Pachner moves}, 
which are the standard ergodic moves that allow to transform into one another any two PL homeomorphic triangulations 
of the same manifold~\cite{Pachner1987}. 

In Section \ref{sec:substructures}, we discuss how RSHT can be used to reach interesting small (or even vertex-minimal) 
triangulations and subcomplexes ``hidden'' inside triangulated manifolds. For the sake of applications, one should declare right away
that RSHT is designed to focus on the (simple-)homotopy, and not the homeomorphism type. So in case we start with a collection 
of points in 10-space, say, which all lie ``approximately'' on a M\"obius strip, the effect of performing RSHT on the \v{C}ech complex 
of the point set would be to detect an~$S^1$, and not a M\"obius strip. Yet, RSHT is capable of detecting, for example, closed surfaces
or higher-dimensional closed manifolds in data, beyond just determining their homologies.
 
 \enlargethispage{3mm}
 
It takes considerable effort to build examples of contractible complexes for which RSHT does \emph{not} practically succeed 
in revealing contractibility, if interrupted after a million steps, say. Respective examples, showcased in the last Section \ref{sec:AK} of our paper, 
are based on the Akbulut--Kirby $4$-spheres \cite{akbulut1985} and triangulations thereof \cite{tsuruga2013}. The homeomorphism type of these 
``tangled'' triangulations of $S^4$ is notoriously difficult to recognize.

\section{Pure elementary expansions}

Any two simple-homotopy equivalent complexes are homotopy equivalent. The converse is true for complexes whose fundamental group 
has trivial \emph{Whitehead group} (see \cite{CohenSimple} or \cite{Mnev} for the definition), but  false in general: 
Counterexamples in dimension two can be obtained by triangulating the cell complexes in~\cite{Lustig},
while counterexamples in dimension three or higher had been known to exist long before \cite{Milnor}.

It is an easy consequence of the theory of Gaussian elimination for integer matrices that the Whitehead group of the trivial group is trivial. 
Therefore, any two homotopy-equivalent simply-connected complexes are also simple-homotopy equivalent. 
More generally, it is known that the Whitehead group of a group $G$ is trivial if $G$ is 
\begin{compactitem}
\item $\mathbb{Z}$ \cite{Higman}, $\mathbb{Z} \oplus \mathbb{Z}$  \cite{BassHellerSwan}, and more generally, any free Abelian group  \cite{BassHellerSwan}, 
\item any of the cyclic groups $\mathbb{Z}_2$, $\mathbb{Z}_3$, $\mathbb{Z}_4$, $\mathbb{Z}_6$ \cite{CohenSimple},  
\item any subgroup of the braid group $B_n$ \cite{FarrellRouchon}, or of any Artin group of type $A_n$, $D_n$, $F_4$, $G_2$, $I_2(p)$, $\tilde{A}_n$, $\tilde{B}_n$, $\tilde{C}_n$, or $G(de, e, r)$ ($d,r \ge 2$) \cite{roushon2020certain};
\item any free product of groups listed above, so in particular  $\mathbb{Z} \ast \mathbb{Z}$ or any free group \cite{Stallings};
\item and further cases \cite{grenier2007triviality}; in fact, the Farrell--Jones conjecture implies that any torsion-free group should appear in the present list \cite{LuckEtAl}.

\end{compactitem}

 Any two homotopy-equivalent complexes whose fundamental group appears in the list above are of the same simple-homotopy type. 

Whitehead's work allows us to be more specific on the dimension (although not on the number) of the intermediate complexes involved in the definition 
of simple-homotopy equivalence, as follows. An elementary collapse is called an \emph{$i$-collapse} if the dimension of the two faces removed are $i-1$ and $i$.
Similarly, an \emph{$i$-anticollapse} is one that adds a pair of faces of dimension~$i-1$ and $i$.  The \emph{order} of a formal deformation
will be the maximum $i$ for which $i$-collapses or $i$-anticollapses are involved in the sequence.


\begin{theorem} [Whitehead {\cite[Theorems 20 \& 21]{whitehead1939simplicial}}] \label{thm:Whitehead}
Let  $K$ and $L$ be two homotopy-equivalent simplicial complexes. If the Whitehead group of their fundamental group is trivial, 
then there is a formal deformation from $K$ to $L$ whose order  does not exceed $\max \{ \dim K, \dim L \} + 2$. 
If, in addition, $L$ is a deformation retract of $K$, and $\dim K > 2$, then  there is a formal deformation from $K$ to $L$ whose order  does not exceed $\dim K+1$.
\end{theorem}
 
The conjecture that the last statement of the previous theorem might hold also for the case $\dim K = 2$ goes under the name of ``Generalized Andrews--Curtis conjecture'', 
and represents a major open problem in combinatorial topology. It is, however, generally believed to be false \cite{Metzler1993}.
   
Based on Whitehead's work, we would now like to perform ``random anticollapses''. Yet, if we wish to add a  $(d+1)$-dimensional face $\sigma$ to $K$ 
in an elementary anticollapse, then all $d$-faces of $\sigma$ need to be present in $K$ already, except for a single $d$-face~$\tau$. 
However, it is not difficult to construct contractible $d$-complexes $K$ that do not allow any $(d+1)$-anticollapses; cf.~\cite{lofano2019}. 
In these cases, lower-dimensional faces need to be added first. Computationally, this brings an extra difficulty to the introduction of a random model.  
To bypass this difficulty, we adopt a different set of moves.

\begin{definition}
Let $K$ be a $d$-dimensional complex. A \emph{pure elementary $(d+1)$-expansion} is the gluing of a $(d+1)$-dimensional simplex $\sigma$ to $K$ 
in case $\sigma$ intersects $K$ in a $d$-ball. 
\end{definition}

A pure elementary $(d+1)$-expansion combines together in a single move one \mbox{$(d+1)$}-anticollapse plus all the lower-dimensional anticollapses 
that have to be performed first. Hence a sequence of pure elementary expansions  and elementary collapses can be rewritten as a formal deformation. 
Whenever we run out of collapsing steps, we perform exactly one pure elementary $(d+1)$-expansion, and then switch back to elementary collapses.  
When the complex is reduced to a  point, we stop.

\section{Implementation of Random Simple-Homotopy}
\label{sec:implementation}

\begin{algorithm}[th]\DontPrintSemicolon
  \caption*{Random Simple-Homotopy}
  \label{algo}
  \KwIn{simplicial complex $K$}
  \KwOut{simplified simplicial complex}
  \While {$\dim(K) \neq 0$ \textbf{and} $i < $max\textunderscore step}{
        \While {$K$ has free faces} {
            \textbf{(C)}: perform a random elementary collapse
            }
        \eIf {$\dim(K) =d \neq 0$  \textbf{and} there are induced pure $d$-balls on $d+2$ vertices} {
            \textbf{(E)}: perform a random pure elementary $(d+1)$-expansion \;
            \textbf{(CC)}: perform an elementary collapse deleting the newly added \\
            \mbox{}\hspace{10.625mm} $(d+1)$-face and one of its $d$-faces that was already in $K$ \;
            }
        {\textbf{(S)}: perform (E) + (CC) on a $d$-facet with $d+1$ vertices\\
            i++ \;
            }
  }
  \Return $K$ \;
  \end{algorithm}

Algorithm RSHT provides a description of the Random Simple-Homotopy procedure in pseudocode. The actual implementation can be found on GitHub at \cite{RHExt} as a \texttt{polymake} \cite{polymake:2000} extension.
It is based on the two different types of basic operations: random \emph{collapses} (C) and random \emph{pure elementary expansions} (E) 
plus collapsing steps (CC) that ensure that a pure elementary expansion is not undone immediately by the next regular
collapsing step (C). The step (S) allows facet \emph{subdivisions} in case no other pure elementary expansions are available.

Random collapses (C) are discussed as part of Random Discrete Morse Theory in \cite{benedetti2014random}.
A fast implementation of random collapses in \texttt{polymake} is described in \cite{joswig2014}.
Hence, it remains to implement random pure elementary expansions (E).

While collapses in \texttt{polymake} can be carried out fast in the Hasse diagram of $K$, there is no explicit implementation in \texttt{polymake} 
to expand the Hasse diagram of $K$ to include the faces of a $(d+1)$-simplex $\sigma$ that is added in a pure elementary expansion.
Thus, for every pure expansion we recompute the Hasse diagram for $K+\sigma$ and then proceed with random collapses 
in the new Hasse diagram of $K+\sigma$. For various input examples of non-collapsible, contractible complexes, 
relatively few pure expansions are needed (see Sections~\ref{sec:Dunce-Bing} and \ref{sec:substructures}); 
thus the extra cost for recomputing Hasse diagrams stays low.

\begin{remark}
Pure elementary expansions are not the only operations to modify a given complex $K$ by expanding it. Another more general possibility 
would be to glue additional $(d+1)$-simplices to $K$ along induced contractible subcomplexes (of mixed dimension).  
This provides more options to modify~$K$, but at the price of having to check subcomplexes for contractibility.
As we experienced after running various experiments, this seems expensive without any advantage. 
We therefore decided to stick to pure elementary expansions. In fact, checking whether an induced subcomplex 
on $d+2$ vertices is a pure $d$-ball is very fast: It can be achieved by a lookup in the Hasse diagram.
\end{remark}

\begin{remark}
By Whitehead's Theorem~\ref{thm:Whitehead}, we might be forced to first go up by two dimensions (and not just by one
as we do in Algorithm RSHT) to find a formal deformation from a complex $K$ to some homotopy equivalent complex $L$.
This could be incorporated in the algorithm by performing not only single pure elementary $(d+1)$-expansions
followed immediately by collapses, but by allowing sequences of pure elementary $(d+1)$-expansions followed
by pure elementary $(d+2)$-expansions before switching back to collapses. In principle, this generalized procedure 
could be set up in a simulated annealing fashion, in a completely analogous way to what we do here; but for the examples we study in the subsequent 
Sections~\ref{sec:Dunce-Bing} and \ref{sec:substructures}, we shall restrict ourselves to the basic algorithm RSHT, as this already works well.
\end{remark}

\section{Bistellar flips and artefacts}
\label{sec:bistellar}

Pure elementary $(d+1)$-expansions have (at least for $d$-manifolds in low dimensions $d\leq 6$) a clear interpretation in  terms of bistellar flips. 
In fact, let $K$ be a $d$-complex. In a pure elementary $(d+1)$-expansion, 
some $(d+1)$-simplex $\sigma$ is glued to $K$ along a $d$-ball consisting of $1\leq k\leq d+1$ of the $d$-faces of $\sigma$; 
let $r$ be the intersection of these $k$ faces. If $r$ is contained in no further $d$-face of $K$, then after adding $\sigma$, collapsing it away 
with one of the $k$ $d$-faces, and collapsing further lower-dimensional faces, we are left with a complex $K'$ that is obtained from $K$ 
via a bistellar move; cf.~\cite{bjorner2000}. If instead $r$ is contained in more than $k$ $d$-faces of $K$, 
then in passing from $K$ to $K'$ the  facet degree of $r$ is decreased by one. 

\begin{figure}[!t]
\centering
 \includegraphics[width=0.6\linewidth]{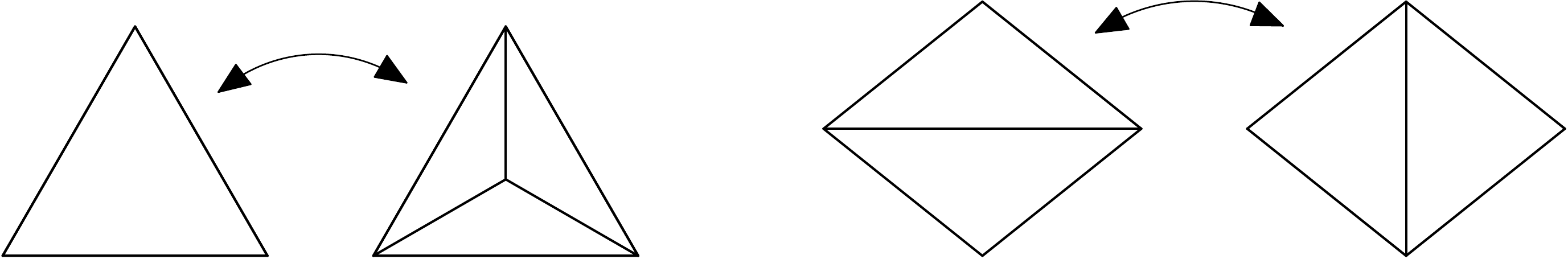}
 \captionsetup{width=.85\linewidth}
 \caption{Expansions as bistellar flips.} 
 \label{fig:Bist}
\end{figure}

\begin{example} \label{ex:bistellar}
If we glue a tetrahedron $\sigma$ to a $2$-complex $K$ along a $2$-disk in $\partial \sigma$, the disk can either consist of 1, 2, or 3 triangles.
In the first case, the complex $K'$ resulting after the collapses is a subdivision of $K$. (The triangle $\tau$ of $K$ is subdivided 
using the unique vertex of $\sigma$ not in~$K$; see Figure~\ref{fig:Bist}, left.) In the second case, if $r$ is the edge common 
to the two triangles of $\partial \sigma$ in which $\sigma$ intersects $K$ and $r$ is contained in exactly these two triangles of $K$,  
then $r$ is flipped to yield $K'$; see Figure~\ref{fig:Bist}. In the third case, the transition from $K$ to $K'$ ``undoes'' a subdivision. 
\end{example}

\begin{figure}
\centering
  \includegraphics[width=4.3cm]{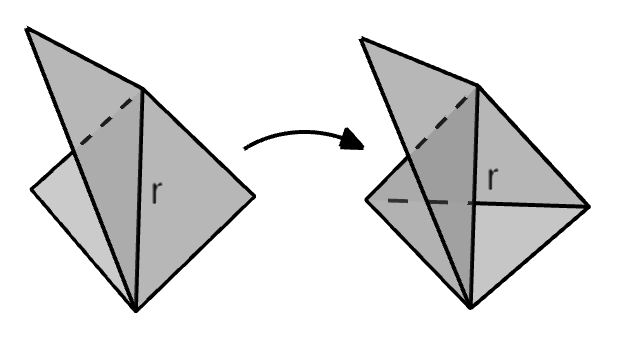}
 \captionsetup{width=.85\linewidth}
 \caption{Reduction of the face degree.}
 \label{fig:PureExp}
\end{figure}

\begin{example} \label{ex:reduction_degree}
Let $K$ be $2$-dimensional and let $\sigma$ be a tetrahedron glued to $K$ along two triangles whose intersection is $r$, 
and suppose that this $r$ is contained in exactly three triangles of $K$.  Then after the addition of $\sigma$ and its removal, 
$r$ will be contained in two of the triangles of~$K'$; see Figure~\ref{fig:PureExp}.
\end{example}

If, in a pure elementary $3$-expansion, some tetrahedron is glued on top of two adjacent triangles $\varrho_1, \varrho_2$ of a triangulated $2$-manifold, 
then, after collapsing away the tetrahedron  together with  $\varrho_1$, the resulting triangulation will still contain $\varrho_2$
and (as a free face) the edge $e = \varrho_1 \cap \varrho_2$. This edge $e$ is thus the only free ($1$-)face; hence, it will be selected
in the incoming (C) step of RSHT. As a result, the combination (E) + (CC) + (C) is a proper bistellar flip---and the diagonal of the two initial triangles gets flipped. 
In the case of a subdivision, the combination (E) + (CC) is a proper bistellar flip as well. Thus, it remains to inspect the case when a subdivision is undone. 
After the addition of a tetrahedron (E) and the deletion of one of the initial three triangles along with the tetrahedron in the (CC) step, 
the other two initial triangles remain, and we have (two) free edges for two further (C) steps. In contrast to the previous cases,  after the two (C) steps, 
the resulting triangulation is not a surface yet---as we still have the intersection vertex of the three initial triangles
as a free vertex that is connected to the modified triangulated surface by an edge. That is, the result of (E) + (CC) + (C) + (C)
is a triangulated surface with an additional edge sticking out. This edge is then collapsed away in another (C) step.

This situation generalizes as follows:

\begin{lemma}\label{lem:BF}
Let $K$ be a triangulation (without free faces) of a $d$-manifold~$M$ and suppose that the $(d+1)$-simplex $\sigma=[0,1,\dots,d+1]$ intersects $K$
 in a pure $d$-ball $B$ with $1\leq k \leq d+1$ $d$-facets on the $d+2$ vertices $0, 1, \dots, d+1$ of $\sigma$ 
so that (w.l.o.g.)\ $B=\mbox{$[0,1,\dots,d-k+1]$}*\partial [d-k+2,d-k+1,\dots,d+1]$. We add $\sigma$ (and its faces) to~$K$ and, by step (CC) of RSHT, 
ban those facets of $\sigma$ as free faces that do not contain $[0,1,\dots,d-k+1]$. 
\begin{compactitem}
\item If $k\leq 7$,  then running RSHT on $K\cup\sigma$ until no further free faces are available yields a triangulation $K'=K-B+B'$ of $M$ 
with \[B'=\partial[0,1,\dots,d-k+1]*\mbox{$[d-k+2,d-k+1,\dots,d+1]$},\] i.e., $K'$ is obtained from $K$ by a bistellar flip. 
\item If $k > 7$ (which can occur for $d>6$ only), then running RSHT on $K\cup\sigma$ until no further free faces are available 
might terminate in a non-pure simplicial complex $K''$ that is the union of a triangulation of $M$ with 
a contractible, but non-collapsible lower-dimensional complex on the vertices $d-k+2,d-k+1,\dots,d+1$.
\end{compactitem}
\end{lemma}

\begin{proof}
Step (CC) of RSHT implies that our first collapsing move will remove a facet of $B$ along with the added $(d+1)$-simplex $\sigma$. 
At any consecutive collapsing step~(C), the faces involved in the collapses will be of the form $[0,1,\dots,d-k+1]* \tau$, 
where $\tau \in \partial [d-k+2,d-k+1,\dots,d+1]$ (because our starting complex $K$ had no free faces). 
The restriction of the collapsing sequence to $\partial [d-k+2,d-k+1,\dots,d+1]$ gives us  a valid collapsing sequence 
of the simplex $[d-k+2,d-k+1,\dots,d+1]$, where the first collapsing move is induced by the initial step (CC). Now:
\begin{compactitem}
\item If $k\leq 7$, the simplex $[d-k+2,d-k+1,\dots,d+1]$ has at most seven vertices; and by \cite{bagchi2005combinatorial}  
every contractible simplicial complex with $k\leq 7$ vertices is collapsible, i.e., the collapsing sequence induced by RSHT on $[d-k+2,d-k+1,\dots,d+1]$ 
will terminate at a single point. It follows that $K'=K-B+B'$.
\item If $k>7$, then the collapsing sequence on $[d-k+2,d-k+1,\dots,d+1]$ might get stuck on a contractible,
but non-collapsible subcomplex of dimension at least two  \cite{Crowley,lofano2019}, 
and thus the resulting complex $K''$ need not be pure. \qedhere
\end{compactitem} 
\end{proof}

Note that in the special case when $d=7$ and $k=8$ we might get stuck on a subcomplex $[0]*D\subseteq [0]*\partial[1,\dots,8]$,
with $D$ an $8$-vertex triangulation of the Dunce Hat; cf.\ \cite{benedetti2009dunce}. The resulting complex $K'=K-B+B'+[0]*D$
then deviates from the modification via a bistellar flip, $K-B+B'$, by the additional cone $[0]*D$ with apex $[0]$
over the (contractible) Dunce Hat $D$ in the $2$-skeleton of $\partial[1,\dots,8]$. The complex $K'=K-B+B'+[0]*D$
deformation retracts to $K-B+B'$, but has no free faces. 

\begin{theorem}[Reduction of pure elementary $(d+1)$-expansions to bistellar flips] \label{thm:bistellarflips}
Let $K$ be a triangulation of a $d$-manifold $M$ with $d\leq 6$. Any pure elementary $(d+1)$-expansion 
followed by collapses (as long as free faces are available) induces a bistellar flip on $K$.
\end{theorem}

\begin{proof}
The statement follows from Lemma \ref{lem:BF} and the fact that the maximum number of facets of a pure $d$-ball on $d+2$ vertices is $d+1$.
\end{proof}

\begin{corollary}[Manifold stability]
Let $K$ be a (not necessarily pure) simplicial complex. If we run RSHT on $K$ and at some point reach a simplicial complex $K'$
that triangulates a $d$-manifold with $d\leq 6$, then from then on, whenever there are no free faces in the further run of RSHT, 
the respective temporary complex $\tilde{K}$ is a $d$-manifold as well, and $\tilde{K}$ is bistellarly equivalent to $K'$.
\end{corollary}

To avoid lower-dimensional artefacts $[0,1,\dots,d-k+1]*N\subseteq [0,1,\dots,d-k+1]*\partial [d-k+2,d-k+1,\dots,d+1]$
in the modification $K'=K-B+B'+[0,1,\dots,d-k+1]*N$ of a triangulated manifold $K$, involving a contractible, non-collapsible complex $N$
for $d\geq 7$ and $k\geq 8$, we should switch to bistellar flips $K'=K-B+B'$ once we know that $K$ is a manifold. Quite often, 
this is not clear a priori---in fact, testing whether $K$ is a manifold is an undecidable problem for $d\geq 6$; cf.\ \cite{joswig2014}. 

In practice  \cite{joswig2014}, on a $7$-simplex it is nearly impossible to get stuck with random collapses. On the $8$-simplex,
only about 0.0000012\% of the runs of random collapses get stuck. But in higher dimensions, the situation changes dramatically: For example, for the $25$-simplex, 
contractible but non-collapsible substructures are encountered in 92\% of the runs.

Another option to deal with the artefacts would be to run RSHT on lower-dimensional parts to ``melt away'' the artefacts. 
However, in our experiments in Sections~\ref{sec:Dunce-Bing} and~\ref{sec:substructures} we only focus 
on top-dimensional pure elementary expansions, since the terminal triangulations of the examples we consider are all of dimension $d\leq 6$.

In case a general complex $K$ has no free faces and is not a manifold, then a sequence (E) + (CC) + (C) + \dots + (C) 
until no further collapses are possible might reduce $K$ in dimension or can reduce (or increase) the degree of a face in $K$,
as we have seen in Example~\ref{ex:reduction_degree} and Figure~\ref{fig:PureExp}. In the latter case,
we can regard the sequence as a \emph{generalized bistellar flip}. These generalized operations give flexibility 
in the modification of a given complex $K$.

 \subsection{Selection of expansions and simplification of complexes}

We next discuss in more detail how the pure elementary expansions are selected and why Algorithm RSHT has a tendency to simplify simplicial complexes
to yield small or even vertex-minimal triangulations.  First, we note that RSHT, apart from temporarily adding \mbox{$(d+1)$}-faces 
in the pure elementary expansion steps (E),  never  increases the dimension of the complex. 

As outlined in the introduction, for any $d$-facet of a $d$-dimensional complex $K$, chosen uniformly at random,
we can check for all neighboring $d$-facets whether the induced subcomplexes on the combined $d+2$ vertices are pure $d$-dimensional.
From the collection of all available such pure induced $d$-balls on $d+2$ vertices, we pick one uniformly at random for a pure elementary
$d$-expansion step~(E). However, in general, such pure induced $d$-balls on $d+2$ vertices need not exist.
For example, in the case of neighborly triangulations of surfaces, the induced subcomplexes on the four vertices of two adjacent
triangles are the two triangles \emph{plus} the opposite diagonal edge; such subcomplexes are not contractible.
In such a case, the only possible pure elementary expansion is by picking a facet (uniformly at random) 
as a pure $d$-ball and initiating a subdivision~(S). An example of a triangulated $3$-sphere on $16$ vertices
that allows no bistellar flips (apart from subdivisions of tetrahedra) is given in \cite{dougherty2004}.

\enlargethispage{2mm}

\begin{lemma}
Let $K$ be a triangulated circle $S^1$ with $n>3$ vertices. Then $K$ is reduced by Algorithm RSHT to the boundary
of a triangle in $n-3$ pure elementary expansion steps~(E), each followed by two collapsing steps (CC) + (C).
\end{lemma}

In the case of triangulations of $S^2$ with $n>4$ vertices, there always are admissible edge flips, and thus Algorithm RSHT never adds 
a vertex in a subdivision step (S). A~vertex can get removed in the reversal of a subdivision once the current triangulation has a vertex of degree $3$.
However, the boundary of the octahedron has all of its vertices of degree $4$; in fact, there are infinitely many triangulations of $S^2$ 
with all vertex degrees at least four. In any such example, the removal of a vertex is not immediately possible. 
But after a suitably long sequence  of random edge flips, eventually vertices of degree $3$ show up, and the three incident triangles 
to such a vertex have the chance to get chosen for an induced pure $2$-ball to remove the vertex of degree $3$.

Similarly, general complexes $K$ are simplified and reduced in size by collapsing away collapsible parts and by reversing subdivisions 
to reduce the number of vertices---but without a universal guarantee for success (as contractibility is undecidable).

\section{Classical examples}
\label{sec:Dunce-Bing}

In this section, we test how the Algorithm RSHT performs on the Dunce Hat, on Bing's House with two rooms, 
and on similar, ``classical'' examples of contractible complexes. 
It turns out that the number of pure elementary expansions needed to reduce these complexes to a single vertex is conveniently low: one pure elementary expansion suffices 
for an $8$-vertex triangulation of the Dunce Hat; five pure elementary expansions suffice for a simplicial version of Bing's house with two rooms;  
and in general, six tetrahedra are sufficient to collapse Bing's house with $k$ rooms (\textbf{Theorem \ref{thm:Bingk}}). 
Triangulations of these examples can be found online at the ``Library of Triangulations'' \cite{library}.

\begin{figure}[!htb]
    \centering
    \begin{subfigure}{.31\textwidth}
        \centering
        \includegraphics[width=\linewidth]{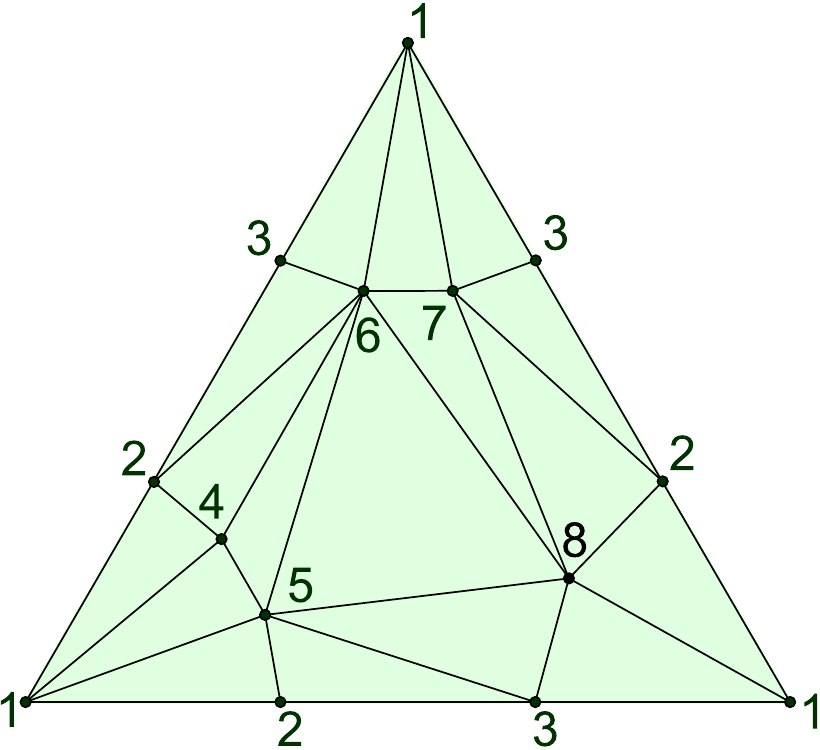}
        \caption{An $8$-vertex triangulation of the Dunce Hat}
        \label{DH:1}
    \end{subfigure}%
    \hspace{0.02\textwidth}
    \begin{subfigure}{.31\textwidth}
        \centering
        \includegraphics[width=\linewidth]{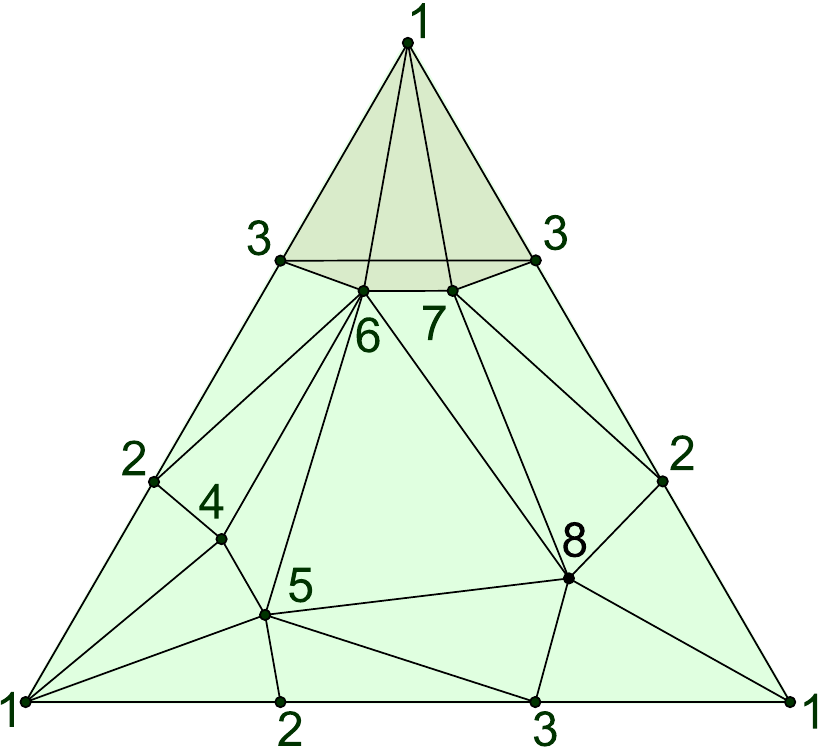}
        \caption{Anticollapsing the tetrahedron 1367.}
        \label{DH:2}
    \end{subfigure}%
    \hspace{0.02\textwidth}
    \begin{subfigure}{.31\textwidth}
        \centering
        \includegraphics[width=\linewidth]{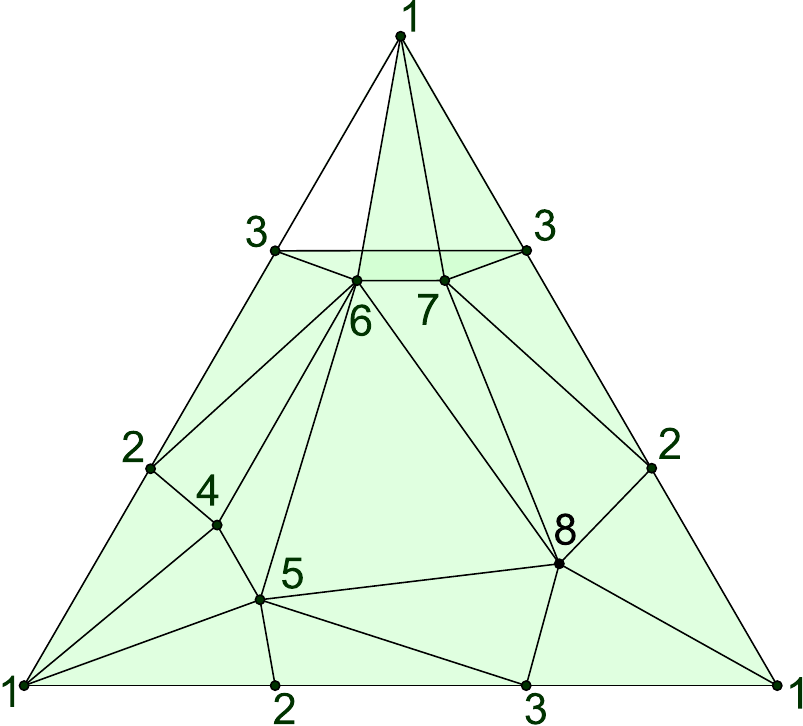}
        \caption{Collapsing away the tetrahedron 1367.}
        \label{DH:3}
    \end{subfigure}
    \caption{A formal deformation of the Dunce Hat.}
\end{figure}

\subsection{The Dunce Hat}

The Dunce Hat \cite{Zeeman64Dunce} is the most famous example of a contractible, but non-collapsible complex; compare \cite{benedetti2009dunce}. 
It is obtained by glueing together the three edges of a single triangle in a non-coherent way. The Dunce Hat can be triangulated as a simplicial complex 
with eight vertices (see Figure~\ref{DH:1}); and eight vertices is fewest possible, as every contractible simplicial complex on seven vertices 
is collapsible \cite{bagchi2005combinatorial}. No triangulation of the Dunce Hat is collapsible, since there are no free edges to start with.

The Dunce Hat of Figure \ref{DH:1} admits two anticollapsing moves, the addition of the tetrahedron 1245 or alternatively the addition of the tetrahedron 1367. 
In Figure~\ref{DH:2}  we added 1367. All of the triangles in 1367 are free, 
since this is now the only tetrahedron present. If we collapse away the triangle 367, we recover the initial complex of Figure~\ref{DH:1}. 
If instead we choose to delete the free triangle 136, we obtain the triangulation displayed in Figure~\ref{DH:3}. This triangulation 
has a free edge, 16, that allows us to get rid of the triangle 167. After this elementary collapse, the edge 17 becomes free,
allowing us to remove the triangle 137. But now the edge 13 is free, and it can easily be seen
that the deletion of the triangle 138 paves the way to a full collapse down to a single vertex. 

\begin{lemma}
One pure elementary $3$-expansion suffices to reduce to a vertex the $8$-vertex triangulation of the Dunce Hat from Figure~\ref{DH:2}(a).
\end{lemma}

In $10^4$ runs, RSHT used on average 2.4145 pure elementary $3$-expansions to reduce the $8$-vertex 
Dunce Hat to a point; see Section~\ref{sec:contractible} and Table~\ref{tbl:examples}.

\subsection{The Abalone}

The Abalone \cite{Metzler1993}, sometimes called \emph{Bing's House with one room}, is another example of a contractible but non-collapsible complex. 
We are not aware of any triangulation of this space in the literature, so we present one, \texttt{Abalone}, with 15 vertices: 

\medskip
\begin{tabular}{llllllll}
1\,2\,7 & 1\,2\,9 & 1\,3\,8 & 1\,3\,9 & 1\,4\,7 & 1\,4\,8 & 1\,4\,9 & 2\,3\,7   \\
2\,3\,15 & 2\,9\,15 &  3\,7\,8 & 3\,9\,14 & 3\,14\,15 & 4\,5\,7 & 4\,5\,8 &  4\,6\,7 \\
4\,6\,9  & 5\,6\,9 & 5\,6\,10 &  5\,7\,10 & 5\,8\,9 & 6\,7\,11 & 6\,10\,11 & 7\,8\,10 \\
7\,8\,11 & 8\,9\,12 & 8\,9\,13 &  8\,10\,12 &  8\,11\,13 & 8\,12\,13 & 9\,12\,14 & 9\,13\,15  \\
10\,11\,12 & 11\,12\,13 & 12\,13\,14 & 13\,14\,15.
\end{tabular}
\medskip\\

\begin{figure}[!htb]
\centering
 \includegraphics[width=7.5cm]{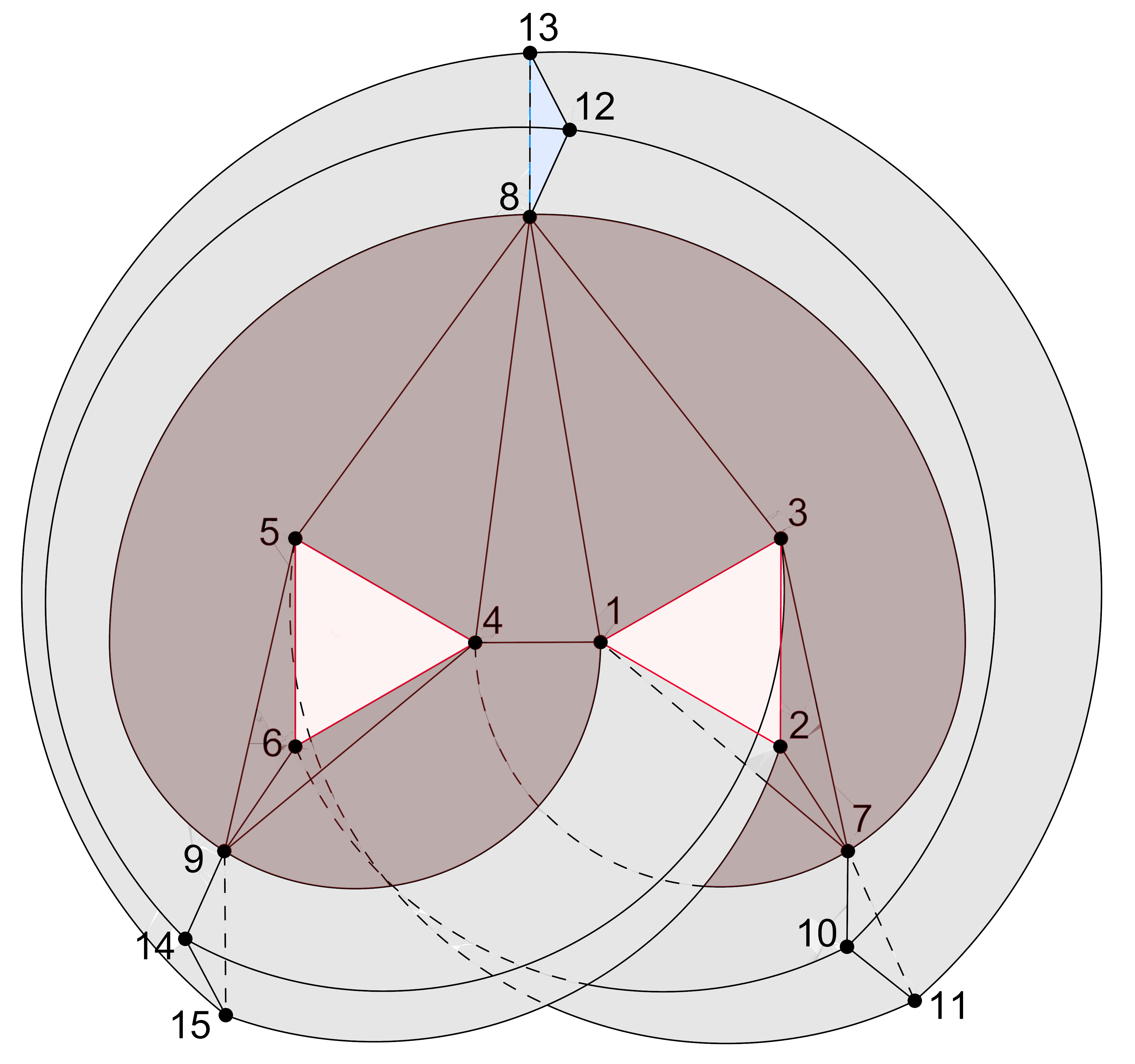} \hfill
 \includegraphics[width=6.2cm]{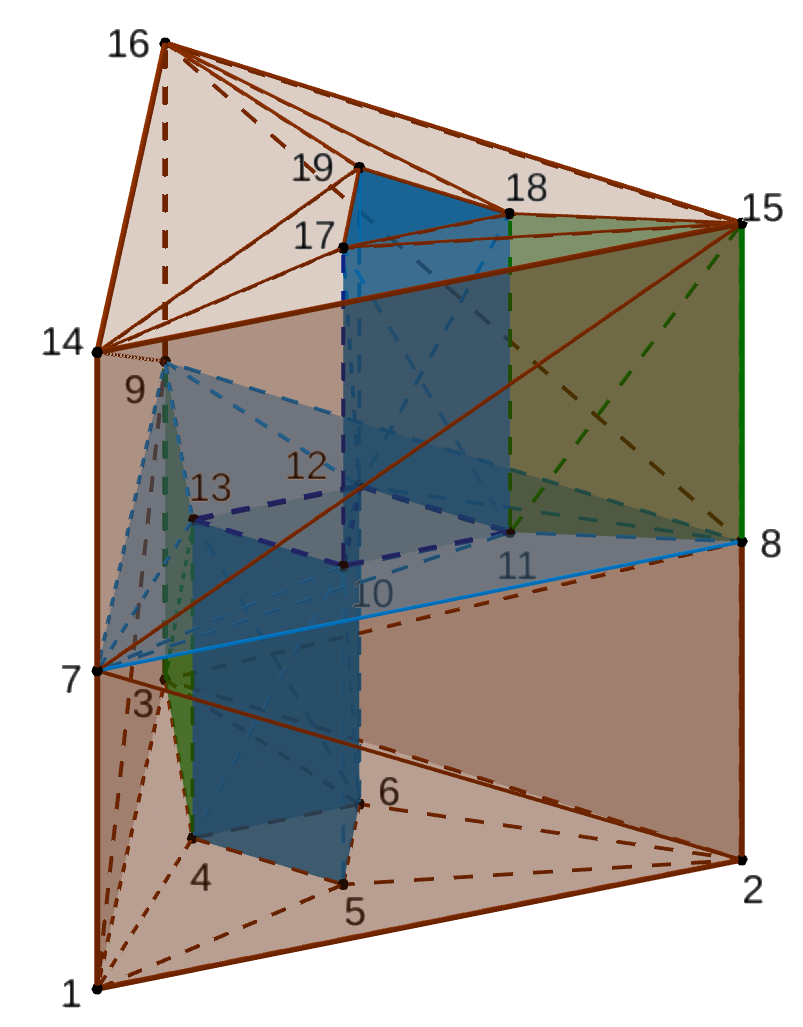}
 \captionsetup{width=.9\linewidth}
 \caption{Triangulations \texttt{Abalone} of the Abalone (left) and \texttt{BH} of Bing's house with two rooms (right).}
 \label{figure:AbaloneBing}
\end{figure}

Figure \ref{figure:AbaloneBing} displays this triangulation, although some diagonals have been omitted for reasons of pictorial clarity. 
Essentially, the triangulation consists of a membrane (in dark) from which two prismatic tunnels (in light) originate 
at the two empty triangles $1\,2\,3$ and $4\,5\,6$; and the tunnels are separated by the highlighted triangle 8\,12\,13. 
The Abalone is contractible as can be seen by filling in the two tunnels.

RSHT can reduce the Abalone to a point using only three expansions. 
One way to do so is to free the edge $8\,9$ of Figure \ref{figure:AbaloneBing} by first adding the three tetrahedra 
$8\,9\,12\,13$, $9\,12\,13\,14$, and $9\,13\,14\,15$, in this order, as anticollapsing moves.
The resulting complex is then collapsible. This can either be verified by hand, or via the \texttt{random\textunderscore discrete\textunderscore morse} algorithm
implemented in \texttt{polymake} \cite{benedetti2014random}: The three tetrahedra fill in the prism between the triangle $8\,12\,13$ and the (formerly empty) triangle $9\,14\,15$.
By collapsing away this prism, the edge $8\,9$ becomes free so that the (dark) membrane around the empty triangle $4\,5\,6$ can be collapsed away, which frees the 
tunnel originating at this empty triangle. Its removal then allows to collapse the remaining disk. 

We can interpret the anticollapsing moves followed by collapses as operations that move the walls of the tunnel 
so that eventually the obstruction to collapsibility vanishes.

\subsection{Bing's House with two rooms} 

Bing's House with two rooms  \cite{Bing1964} is an early example of a contractible space no triangulation of which is collapsible. 
For our purposes, we triangulate Bing's House  as a triangular prism with two floors, two tunnels to reach the floors, 
and all rectangular walls subdivided into two triangles each. Figure~\ref{figure:AbaloneBing} displays the following (small) triangulation $BH$
with $f=(19, 65, 47)$ (with the list of facets also available online as example \texttt{BH} at~\cite{library}):

\medskip
\begin{tabular}{llllllll}
1\,2\,5   &  1\,2\,7   & 1\,3\,4   & 1\,3\,9       & 1\,4\,5   & 1\,7\,9     &  2\,3\,6  &  2\,3\,8  \\
 2\,5\,6  &   2\,7\,8  &  3\,4\,6  &  3\,4\,13  &  3\,8\,9  & 3\,9\,13 &  4\,5\,10  &  4\,6\,13  \\
 4\,10\,13  &  4\,12\,13   &  5\,6\,10  &   6\,10\,12  &  6\,12\,13  &  7\,8\,11  &  7\,8\,15  & 7\,9\,13  \\
 7\,9\,14  &  7\,10\,11 &  7\,10\,13 &  7\,14\,15  &  8\,9\,12  &  8\,9\,16  &  8\,11\,12  &  8\,11\,15  \\
 8\,15\,16  &  9\,12\,13  &  9\,14\,16  & 10\,11\,17  &  10\,12\,17  &   11\,12\,18  &  11\,15\,18  &  11\,17\,18 \\
 12\,17\,19  &  12\,18\,19  & 14\,15\,17  &  14\,16\,19   & 14\,17\,19  &  15\,16\,18  &  15\,17\,18  & 16\,18\,19.
\end{tabular}
\medskip

RSHT is able to reduce Bing's house to a point by means of five (successive) expansions (in the upper room, each followed by collapses
so that the outer walls of Bing's house are moved towards the upper tunnel). Here is a possible strategy.
By successively adding five tetrahedra in the upper room of our Bing's House triangulation, we fill in a cubical prism 
between the horizontal square \mbox{7--8--11--10} of the medium floor and the square 14--15--18--17 of the ceiling.
The first two tetrahedra 7\,8\,11\,15 and 11\,15\,17\,18 can be added independently, and their addition are proper anticollapsing steps.
The third tetrahedron 7\,11\,15\,17 is a pure expansion, and the addition of the two final tetrahedra 7\,10\,11\,17 and 7\,14\,15\,17 
are again anticollapsing steps. The newly introduced cubical prism connects the outer vertical square 7--8--15--14 
with the vertical square \mbox{10--11--18--17} of the upper tunnel. The resulting complex is collapsible; an explicit collapsing sequence proving this claim is detailed below. 

We start from the outside, by perforating the back square 7--8--15--14. 
Then we entirely remove the interior of the cubical prism along with the two triangles 7\,8\,15 and 7\,14\,15 of the back square
and the two triangles  14\,15\,17 and 15\,17\,18 of the top square. The result is an indented Bing's House triangulation 
with two new side triangles  7\,10\,17 and 7\,14\,17. But now the edge 1\,7\,18 has been freed, and we can use it to collapse away 
the subdivided squares of the triangulation one by one. First the square 10--11--18--17 is collapsed away, which frees the 
edge 10\,11. This edge in turn can be used to remove the horizontal square 7--8--11--10, thus freeing the edge 78. 
Next, we remove the squares 1--2--8--7, 2--3--9--8, 1--3--9--7, the vertical wall 3--4--13--9, then all triangles of the lower floor, 
then the lower tunnel, to end up with the indented upper room with empty triangle 10\,12\,13. 
This remaining complex is a triangulated disc and thus collapsible.
         

\mathversion{bold}       
\subsection{Bing's House with $k$ rooms}
\mathversion{normal}

A recent example of a non-collapsible, contractible complex is \emph{Bing's House with three rooms (and thin walls)} by Tancer~\cite{tancer2016recognition}.
He introduced the example as a gadget to prove that the problem of recognizing collapsible complexes is NP-complete.
The basic layout of the example can be found in~\cite{tancer2016recognition}. Here, we give an explicit triangulation $BH(3)$; 
and extend this construction to $k$ rooms,  $BH(k)$, $k\geq 3$.

The starting point for the construction of $BH(3)$ is to have a ground floor with three triangular holes as depicted in Figure~\ref{figure:Bing3rooms}. 
The floor has the  following triangles:

\medskip
\begin{tabular}{llllllll}
 1\,2\,5       &   1\,2\,15     &   1\,4\,6     &   1\,4\,10     &   1\,5\,7        &   1\,6\,7      &  1\,9\,11    &   1\,9\,14   \\
 1\,10\,12   &   1\,11\,12     &   1\,14\,16    &   1\,15\,16      &   2\,3\,5    &   2\,13\,15  &  3\,4\,6   &  3\,5\,6   \\
 4\,8\,10     &  8\,9\,11  &  8\,10\,11   &  9\,13\,14   & 13\,14\,15.
 \end{tabular}
 \medskip
 
 \begin{figure}
\centering
 \includegraphics[height=6cm]{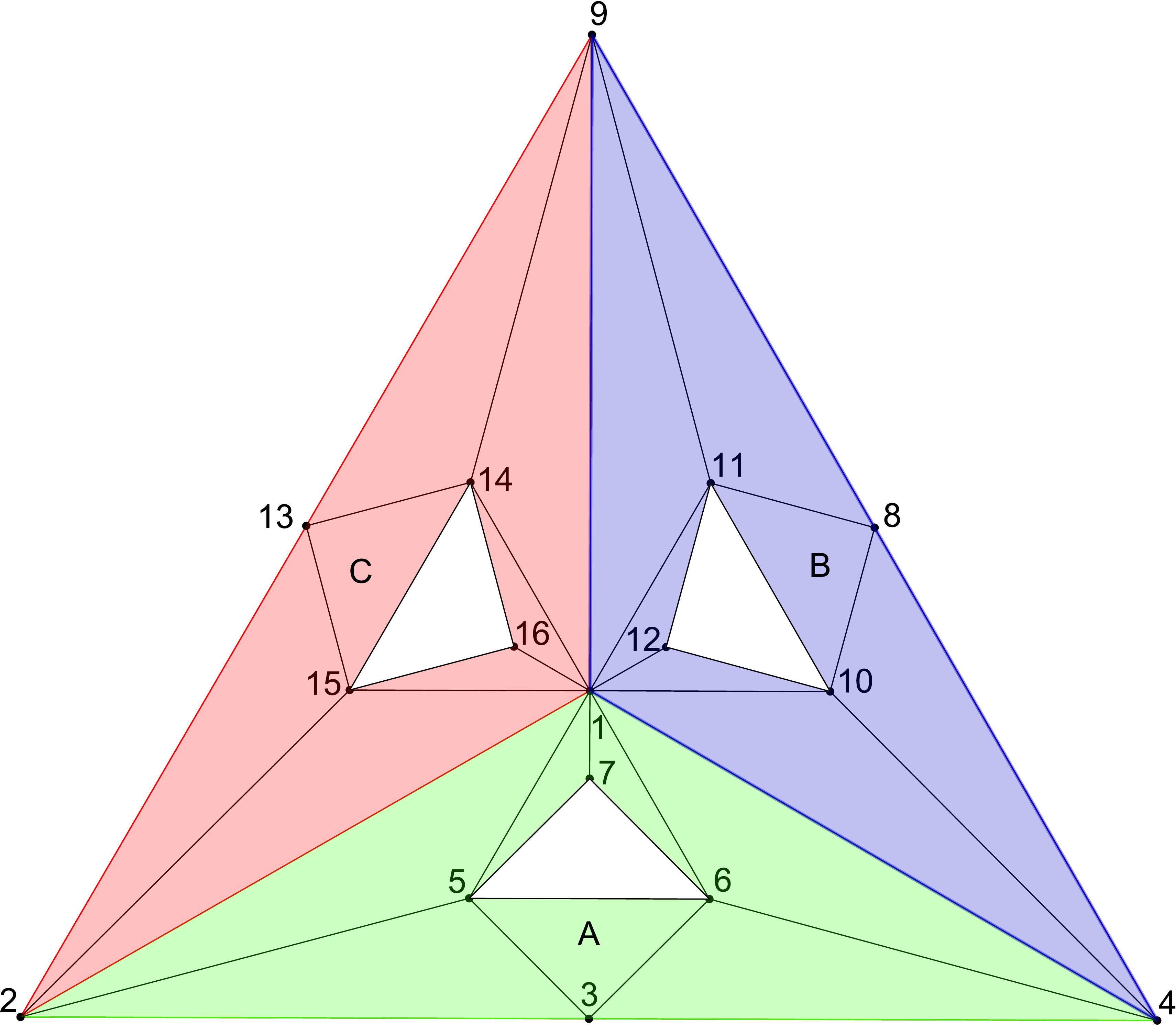} \hfill
 \includegraphics[height=6cm]{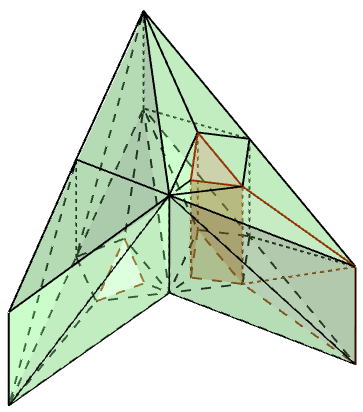}
 \captionsetup{width=.9\linewidth}
 \caption{Ground floor (left) and room $R_2$ (right) of Bing's House $BH(3)$ with three rooms.}
 \label{figure:Bing3rooms}
\end{figure}

Onto the ground floor, we glue three rooms in a coherent way. Room $R_1$ is glued onto the two regions A and B 
and uses nine additional vertices from $17$ to $25$. Room $R_2$, depicted in Figure~\ref{figure:Bing3rooms},  
is glued onto the regions B and C and uses the nine vertices from $26$ to $34$. 
Finally, room $R_3$ is glued onto the regions C and A with further nine vertices ranging from $35$ to $43$.
The rooms $R_2$ and $R_3$ are cyclic copies of the room $R_1$, where $9$ and $18$ are added to the vertex-labels $17$ to $25$ of room $R_1$,
respectively. Concretely, the triangles of room $R_1$ are

\medskip
\begin{tabular}{llllllll}
      1\,2\,17  &  1\,9\,17  &   2\,3\,18  & 2\,5\,18   &  2\,17\,18      &  3\,4\,19   &       3\,18\,19    &    4\,8\,20      \\
      4\,19\,20  &  5\,6\,21      &  5\,7\,21     & 5\,18\,21       &  6\,7\,22   &   6\,21\,22  &  7\,21\,23   &   7\,22\,23  \\
  8\,9\,24  & 8\,20\,24     &  9\,17\,25  &  9\,24\,25     &  17\,18\,21   &    17\,20\,22   &  17\,20\,24  & 17\,21\,23   \\
  17\,21\,23  &  17\,24\,25     &  18\,19\,21   &   19\,20\,22   &  19\,21\,22.
    \end{tabular}
 \medskip
    
\noindent
Those of room $R_2$ are

\medskip
\begin{tabular}{llllllll}
1\,2\,26 & 1\,4\,26 & 2\,13\,33 & 2\,26\,34 &  2\,33\,34 & 4\,8\,27 & 4\,10\,27 & 4\,26\,27 \\
8\,9\,28 &  8\,27\,28 & 9\,13\,29 & 9\,28\,29 & 10\,11\,30 & 10\,12\,30 &  10\,27\,30 &  11\,12\,31 \\
11\,30\,31 &  12\,30\,32 & 12\,31\,32 &    13\,29\,33 & 26\,27\,30 & 26\,29\,31 & 26\,29\,33 & 26\,30\,32 \\
26\,31\,32 & 26\,33\,34 &   27\,28\,30 &  28\,29\,31 & 28\,30\,31,
    \end{tabular}
\medskip

\noindent
and those of room $R_3$ are

\medskip
\begin{tabular}{llllllll}
   1\,4\,35 & 1\,9\,35 & 2\,3\,38 & 2\,13\,37 & 2\,37\,38 &  3\,4\,42 & 3\,38\,42 & 4\,35\,43 \\
   4\,42\,43 & 9\,13\,36 &  9\,14\,36 &  9\,35\,36 & 13\,36\,37 &  14\,15\,39 &  14\,16\,39 & 14\,36\,39 \\
   15\,16\,40 & 15\,39\,40 & 16\,39\,41 & 16\,40\,41 & 35\,36\,39 &  35\,38\,40 & 35\,38\,42 & 35\,39\,41 \\
   35\,40\,41 & 35\,42\,43 & 36\,37\,39 & 37\,38\,40 & 37\,39\,40.  
\end{tabular}
\medskip
\noindent

The three rooms $R_1$, $R_2$, and $R_3$ are then all glued to the upper side of the ground floor. 
Since the vertices of the upper layer of a room are distinct from the vertices
of the upper layers of the other two rooms, there is no conflict for the 
chosen gluing to the same side. To enter the interior of a room, one has to first pass through the tunnel from above
of the room to the left, before the room itself can be entered from below through the
lower left empty triangle.

The previous construction can be generalized to create Bing's Houses $BH(k)$ with $k$ rooms, $k\geq 3$.  
Instead of just three regions, start with $k$ regions that have a triangular hole each, cyclically arranged around a central vertex $1$ on the ground floor,
and attach to it $k$ rooms, $R_1, \ldots, R_k$, in a coherent way, as before. The resulting triangulation has face vector
\[ f=(14k+1, 50k, 36k). \]
A C++-implementation\, \texttt{BH\_k.cc}\, by Lofano to generate the examples $BH(k)$ along with explicit triangulations
\texttt{BH\_3}, \texttt{BH\_4}, and \texttt{BH\_5} can be found online at \cite{library}.

Our next result highlights that in terms of simple-homotopy theory, $BH(k)$ is easy to understand.

\begin{theorem} \label{thm:Bingk}
For any $k\geq 3$, Bing's House with $k$ rooms, $BH(k)$, can be formally deformed to a point using only six pure expansions.
\end{theorem}

\begin{proof}
Since the rooms $R_1,\dots,R_k$ are all identical, we extend to $BH(k)$ the labelling scheme that we used 
for the ground floor and the rooms of $BH(3)$. First we do all the expansions in room $R_1$. By adding the following six tetrahedra 
\[ 2 \, 3 \, 5 \, 18, \:\: 3 \,  5 \, 18 \, 19, \:\:  5 \, 18\,  19 \, 21, \: \: 3 \, 5 \, 6\,  19, \: \:5 \,  6 \, 19 \, 21, \: \:6 \, 19 \, 21\,  22\] 
we fill in the cubical prism between the horizontal square on the vertices 2--3--6--5 of the main floor and the horizontal square 
on the vertices 18--19--22--21 of room~$R_1$'s ceiling. We may now start the collapsing sequence from the outside. 
We  perforate the back square 2--3--19--18 and then remove the whole interior of the prism, along with the back square 2--3--19--18 
and the horizontal square 18--19--22--21 of the ceiling. Now the edge $21 \, 22$ is free. Thus, we can proceed exactly as for 
Bing's House with two rooms: We collapse away the squares 5--6--22--21 and 2--3--5--6, in this order. But now the edge $2 3$ is free; 
so we can use it to collapse away room $R_k$. By induction, we can thus collapse all the rooms one by one. 
 \end{proof}

How does this compare with the experimental results? In $10^4$ runs, RSHT was always able to reduce 
Bing's house with three rooms, $BH(3)$, to a point, using on average about $148$ additional tetrahedra.
In the ``best run'', only $12$ additional tetrahedra were used. For Bing's house with $k$ rooms, $BH(k)$, $4\leq k\leq 7$,
in $10^4$ runs, even in the best case, RSHT tends to perform a growing number of expansions; see Table \ref{tbl:examples}. 
This growing number of used tetrahedra is not surprising, due to the probabilistic model that we used: When selecting from more rooms, 
the number of options for possible expansions gets larger. So if we keep the number of rounds fixed, 
the chances to pick the cleverest sequence of pure expansions will get thinner. 

\begin{table}
\small\centering
\defaultaddspace=0.15em
\caption{RSHT run for a selection of contractible, non-collapsible complexes from \cite{library}.}\label{tbl:examples}
\vspace{-1em}
\begin{tabular*}{\linewidth}{@{\extracolsep{\fill}}r@{\hspace{1mm}}r@{\hspace{1mm}}r@{\hspace{1mm}}r@{\hspace{1mm}}r@{}}
\\\toprule
 \addlinespace
 \addlinespace
 complex  &  $f$-vector &  rounds  & \# expansions & \# expansions \\ 
                &                   &               & (minimum)     & (mean)
  \\ \midrule
   Dunce hat                                 & $(8,24,17)$ & $10^4$ & $1$ & $2.4145$ \\

\addlinespace
\midrule

   Abalone                                   & $(15,50,36)$ & $10^4$ & $3$ & $32.4156$ \\
    
   Bing's House                            & $(19,65,47)$ &  $10^4$  &  $5$ &  $58.0964$             \\
      
   Bing's House 3 rooms, $BH(3)$ & $(43,150,108)$ & $10^4$ & $12$ & $147.9727$ \\
   
   $BH(4)$                                       & $(57,200,144)$ & $10^4$ & $15$ & $167.7727$ \\
   
   $BH(5)$                                      & $(71,250,180)$ & $10^4$ & $27$  & $195.8890$ \\
   
   $BH(6)$                                      & $(85,300,216)$ & $10^4$ & $34$  & $221.2596$ \\
   
   $BH(7)$                                      & $(99,350,252)$ & $10^4$ & $41$  & $244.5763$ \\
   
\addlinespace
\midrule

   \texttt{two\textunderscore optima}  & $(106,596,1064,573)$ &  $10^3$  &  $1$ &   $7.050$            \\ 
   \addlinespace
   \midrule
   Knotted balls &&&& \\
   
   \addlinespace

   Furch's knotted ball & $(380,1929,2722,1172)$  &  $10^3$  &  $1459$  & $1949.950$       \\ 
   
   \texttt{double\textunderscore trefoil\textunderscore ball} & $(15,93,145,66)$ & $10^3$ & $1$ & $29.600$ \\
   
    \texttt{triple\textunderscore trefoil\textunderscore arc} & $(17,127,208,97)$  &  $10^3$  & $6$&    $94.678$  \\ 
        
   \addlinespace
\bottomrule
\end{tabular*}
\end{table}

\section{Experiments on various topologies and substructures}
\label{sec:substructures}

In this section, we explore how our algorithm RSHT performs for further interesting simplicial complexes, whether contractible or not. 
All timings were taken on an Intel(R) Core(TM) i7-4720HQ CPU with 2.60 GHz and 16 GB RAM.

\subsection{Contractible, non-collapsible complexes}
\label{sec:contractible}

Table~\ref{tbl:examples} lists the number of expansions used for the Dunce Hat and Bing's Houses described in the previous section, 
as well as for the contractible complex \texttt{two\textunderscore optima} of \cite{ABL}  and for some knotted balls \cite{lutz2004small, benedetti2013knots}. 
Furch's knotted $3$-ball is the only example in this set for which the runtime is not negligible. In fact, due to the large number of expansions required, 
 it took an average of 85 seconds to complete one round of the algorithm for this $3$-ball.

The explanation of Table~\ref{tbl:simplices} is as follows. If one starts with a single $d$-simplex, with $8\leq d\leq 15$, and one tries to collapse it down to a point,
sometimes one gets stuck in contractible, non-collapsible complexes of intermediate dimension \cite{lofano2019}.  For each initial $d$-simplex we recorded $10$ such examples, 
and on each one of these 10 examples we let  RSHT run for $10^3$ rounds. In each of the rounds, RSHT was able to reduce the respective examples to a point: 
In columns three and four of Table~\ref{tbl:simplices}, we recorded the minimal and average numbers of expansions used. With the increase of the dimension, 
the runtime started to become an issue. For the largest examples, with $d=15$, it took on average around 25 seconds to complete one round.

\begin{table}
\small\centering
\defaultaddspace=0.15em
\caption{RSHT run for contractible, non-collapsible complexes obtained when trying to collapse a $d$-simplex.}\label{tbl:simplices}
\vspace{-1em}
\begin{tabular*}{\linewidth}{@{\extracolsep{\fill}}r@{\hspace{1mm}}r@{\hspace{1mm}}r@{\hspace{1mm}}r@{}}
\\\toprule
 \addlinespace
 \addlinespace
 $d$  & \#  examples $\times$ \# rounds  & \# expansions & \# expansions \\ 
         &                                                      & (minimum)      & (mean)
  \\ \midrule
    
  $8$ &    $10 \times 10^3$   &  $1.0$     &   $1.9310$            \\ 
  $9$ &    $10 \times 10^3$   &  $1.0$     &   $3.6845$            \\ 
$10$ &    $10 \times 10^3$   &  $1.0$     &  $8.4502$            \\ 
$11$ &    $10 \times 10^3$   &  $1.1$     &   $7.6552$            \\ 
$12$ &    $10 \times 10^3$   &  $2.5$     &   $29.4564$            \\ 
$13$ &    $10 \times 10^3$   &  $1.2$     &   $38.3988$            \\ 
$14$ &    $10 \times 10^3$   &  $7.9$     &   $174.7835$            \\ 
$15$ &    $10 \times 10^3$   &  $36.3$    &   $205.1362$            \\ 
    \addlinespace
\bottomrule
\end{tabular*}
\end{table}

\subsection{Submanifolds and non-manifold substructures in manifolds}

If we remove a facet from a triangulation of the $d$-dimen\-sional sphere $S^d$, the resulting simplicial complex is a triangulated $d$-ball, 
and thus has the simple-homotopy type of a point by Whitehead's Theorem~\ref{thm:Whitehead}. In case the initial 
$d$-manifold $M^d$ is not a sphere, the removal of a simplex from a triangulation yields a simplicial complex 
that, depending on $M^d$, may deform to a submanifold or to a non-manifold substructure in $M^d$.
Table~\ref{tbl:manifolds} provides results for some classical examples:
\begin{table}
\small\centering
\defaultaddspace=0.15em
\caption{RSHT run for manifold triangulations minus a facet.}\label{tbl:manifolds}
\vspace{-1em}
\begin{tabular*}{\linewidth}{@{\extracolsep{\fill}}l@{\hspace{1mm}}l@{\hspace{1mm}}l@{\hspace{1mm}}l@{}}
\\\toprule
 \addlinespace
 \addlinespace
  initial complex              &    initial $f$-vector                      &   resulting complex    &  resulting $f$-vector     \\
                                       &                                                   &                                   &                                     \\ 
\midrule
  $\mathbb{R}P^3$ \cite{Walkup}       &    $(11,51,80,40)$                      &   $\mathbb{R}P^2$    &   $(6, 15, 10)$               \\
  $\mathbb{R}P^4$ \cite[ Ch.~3]{lutzdiss}              &    $(16,120,330,375,150)$         &   $\mathbb{R}P^3$    &   $(11,51,80,40)$          \\
  $\mathbb{C}P^2$  \cite{kuhnel19839}     &    $(9,36,84,90,36)$                   &   $S^2$                      &   $(4, 6, 4)$                   \\   
  $\mathbb{H}P^2$ \cite{brehm}                &    $(15,105,455,1365,3003,$      &   $S^4$     &   $(6, 15, 20, 15, 6)$     \\
                                                                   &    $4515,4230,2205,490)$  &                                    &                                             \\
\midrule
  Poincar\'e $3$-sphere \cite{bjorner2000}  &    $(16,106,180,90)$               & $\mathbb{Z}$-acyclic 2-complex            &    $(10, 40, 31)$ \\
    \addlinespace
\bottomrule
\end{tabular*}
\end{table}
Starting with the vertex-minimal triangulation of $\mathbb{R}P^3$ with $11$ vertices, and removing a facet, in $10^4$ runs of RSHT
it took on average $25.2510$ expansions to reach the $6$-vertex triangulation of $\mathbb{R}P^2$. 
From $\mathbb{R}P^4$ to $\mathbb{R}P^3$ it took $885.5957$ expansions. From $\mathbb{C}P^2$ to $S^2$ 
no expansions were used around half of the times; the average number of expansions needed was $2.3543$.
Finally, it took $30.0784$ expansions to reach $S^4$  from $\mathbb{H}P^2$.
For the Poincar\'e homology 3-sphere \cite{bjorner2000}, the RSHT algorithm found a 2-di\-men\-sional $\mathbb{Z}$-acyclic 2-complex on 10 vertices 
(the boundary of the identified dodecahedron) using 2031.732 expansions in less than two minutes per run.

\begin{figure}[!tb]
    \centering
        \begin{subfigure}{.4\textwidth}
            \begin{tabular}{cl}
                \hline
                    torsion & $f$-vector \\
                \hline
                    $\mathbb{Z}_3$ & $(8, 24, 17)$ \\
                    $\mathbb{Z}_4$ & $(8, 26, 19)$ \\
                    $\mathbb{Z}_5$ & $(9, 32, 24)$ \\
                    $\mathbb{Z}_6$ & $(9, 33, 25)$ \\
                    $\mathbb{Z}_7$ & $(9, 34, 26)$ \\
                    $\mathbb{Z}_8$ & $(9, 35, 27)$ \\
                    $\mathbb{Z}_9$ & $(9, 36, 28)$ \\
                    $\mathbb{Z}_{10}$ & $(9,36,28)$ \\
                    $\mathbb{Z}_{11}$ & $(10, 42, 33)$\\
                    $\mathbb{Z}_{12}$ & $(10, 42, 33)$ \\
                    $\mathbb{Z}_{13}$ & $(10,43,34)$ \\
                    $\mathbb{Z}_{14}$ & $(11,50,40)$ \\
                    $\mathbb{Z}_{15}$ & $(11,50,40)$ \\
                \hline
            \end{tabular}
        \end{subfigure}
        \hspace{0.02\textwidth}
    \begin{subfigure}{.53\textwidth}
        \centering
        \includegraphics[height=0.4625\linewidth]{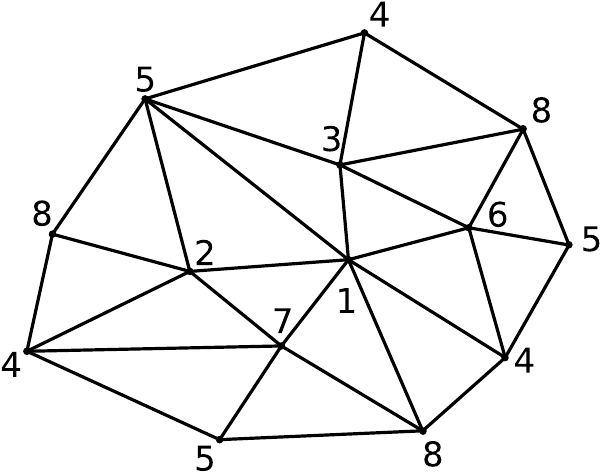}
        \caption{Complex \texttt{d2\_n8\_3torsion} with 3-torsion.}
        \label{figure:3Torsion}
    \end{subfigure}
    \begin{subfigure}{.51\textwidth}
        \centering
        \includegraphics[width=.85\linewidth]{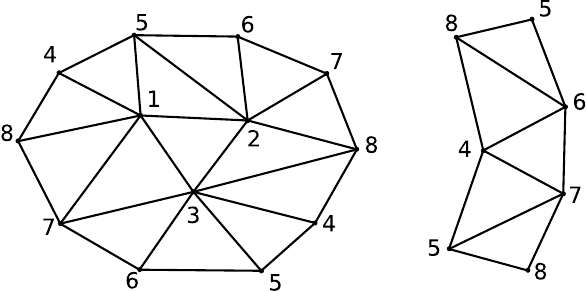}
        \caption{Complex \texttt{d2\_n8\_4torsion} with 4-torsion.}
        \label{figure:4Torsion}
    \end{subfigure}
    \hspace{0.01\textwidth}
    \begin{subfigure}{.46\textwidth}
        \centering
        \includegraphics[width=0.8\linewidth]{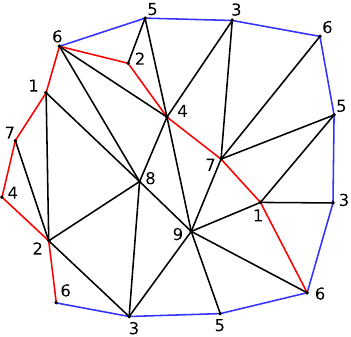}
        \caption{Complex \texttt{d2\_n8\_5torsion} with 5-torsion.}
        \label{figure:5Torsion}
    \end{subfigure}
    \caption{Small substructures with $p$-torsion of the lens spaces $L(p,1)$.}
    \label{Fig:Lens}
\end{figure}
 
The $3$-dimensional lens spaces $L(p,q)$, introduced by Tietze~\cite{tietze}, are well-known topological spaces with torsion in homology. 
Starting from triangulations of the $3$-mani\-folds $L(p,1)$  \cite{brehm1993, Lens} for $p\geq 3$, we aimed for small triangulations 
of $2$-di\-men\-sional simplicial complexes that still have $p$-torsion. (The case $p=2$ has been already considered, since $L(2,1)= \mathbb{R}P^3$.) 
The table in the top left of Figure~\ref{Fig:Lens} gives the $f$-vectors of these smaller complexes; Figure~\ref{Fig:Lens} (a)--(c) shows 
resulting small triangulations \texttt{d2\_n8\_3torsion}, \texttt{d2\_n8\_4torsion}, and \texttt{d2\_n8\_5torsion} (with facets lists available at \cite{library})
with torsion $\mathbb{Z}_3$, $\mathbb{Z}_4$, and $\mathbb{Z}_5$, respectively. The example \texttt{d2\_n8\_3torsion} has the combinatorial symmetry $(2,3)(4,8)(6,7)$;
the example \texttt{d2\_n8\_4torsion} has symmetry $(1,2)(4,6)(7,8)$. In (b), the obtained complex is the union of an $8$-vertex triangulation of the projective plane 
and a M\"{o}bius band. The complex \texttt{d2\_n8\_5torsion} origins from a triangulated disk by identifications highlighted in blue and red.

The following natural problem is open for $p\geq 3$:

\begin{question}
What is the minimal number of vertices $n_{\scriptsize\mbox{min}}(p)$ for a simplicial $2$-complex with $p$-torsion?
\end{question}

An earlier construction of a 2-dimensional simplicial complex with 3-torsion as a \emph{sum complex} on eight vertices
is by Linial, Meshulam and Rosenthal \cite{linial2010sum}. Their example is based on the following collection of subsets 
of $\Z_8$: 
$$X_{\{0,1,3\}}=\{\,\sigma \subset \Z_8 \,:\, |\sigma|=3,\,\, \sum_{x \in \sigma}x\equiv 0,1 \ \mbox{or}\ 3\,\mbox{(mod 8)}\,\}.$$ 
This complex has complete $1$-skeleton and face vector $f= (8,28,21)$.
Three edges of the complex are free, and after collapsing the respective triangles we reach a $2$-complex with $f = (8,25,18)$,
which still has one triangle and one edge more than the example \texttt{d2\_n8\_3torsion}. 
By runninng RSHT on the triangulation with $18$ triangles repeatedly, we again reach \texttt{d2\_n8\_3torsion}---or a second non-isomorphic triangulation 
with the same $f$-vector that is obtained from  \texttt{d2\_n8\_3torsion} by flipping the edge $1$--$5$.

\begin{conjecture}
The examples \texttt{d2\_n8\_3torsion} and  \texttt{d2\_n8\_4torsion} have component-wise minimal
$f$-vectors for complexes with $3$- and $4$-torsion, respectively.
\end{conjecture}

\begin{table}
\small\centering
\defaultaddspace=0.15em
\caption{RSHT run for triangulations of sphere products minus a facet.}\label{tbl:sphereproducts}
\vspace{-1em}
\begin{tabular*}{\linewidth}{@{\extracolsep{\fill}}l@{\hspace{1mm}}l@{\hspace{1mm}}l@{\hspace{1mm}}l@{}}
\\\toprule
 \addlinespace
 \addlinespace
  initial complex              &    initial $f$-vector                      &   resulting complex    &   size of the resulting complex         \\
                                       &                                                   &                                   &                                     \\ 
\midrule
  $S^2 \times S^1$         &  $(12, 48, 72, 36)$                                                 &   $S^2 \vee S^1$       &  $\partial\Delta_3 \cup K^1(3.5382)$   \\
  $S^3 \times S^1$         &  $(15, 75, 150, 150, 60)$                                                 &   $S^3 \vee S^1$       &  $\partial\Delta_4 \cup K^1(7.7617)$   \\
  $S^2 \times S^2$         &  $(16, 84, 216, 240, 96)$                                                 &   $S^2 \vee S^2$       &  $\partial\Delta_3 \cup \partial\Delta_3$   \\
  $S^3 \times S^2$         &  $(20, 130, 420, 710, 600, 200)$                                                 &   $S^3 \vee S^2$       &  $\partial\Delta_4 \cup K^2(11.5460)$   \\
    \addlinespace
\bottomrule
\end{tabular*}
\end{table}

In the description of the torus $S^1\times S^1$ as a square with opposite edges identified, the removal of the interior
of the identified square yields the wedge product $S^1\vee S^1$ of two circles $S^1$ that are glued together at a point.
In general, if we remove a facet from a triangulation of a sphere product, the resulting complex is simple-homotopy
equivalent to the wedge product of the constituting spheres. In the case of $S^2 \times S^1$, the wedge product $S^2 \vee S^1$ 
is of mixed dimension. Since in the implementation of RSHT our focus is on the top-dimensional faces, RSHT is not 
further touching lower-dimensional parts once these are reached via collapses. Thus, the resulting triangulations of $S^2 \vee S^1$
are of the form $\partial\Delta_3 \cup K^1$, consisting of the vertex-minimal triangulation of $S^2$ as the boundary complex $\partial\Delta_3$
of a $3$-simplex $\Delta_3$ union a $1$-dimensional complex $K^1$. 

Depending on the intersection of $K^1$ with $\partial\Delta_3$, $K^1$  either is a path (a $1$-dimensional ball) or a loop (a $1$-sphere $S^1$). 
For a unified description in Table~\ref{tbl:sphereproducts}, we write $K^1(4.5382)$ to point out that $K^1$ has (in $10^4$ runs of RSHT) 
on average $4.5382$ edges. Table~\ref{tbl:sphereproducts} gives results for further sphere products, where for the lower-dimensional parts 
the average number of facets are listed. The initial triangulations of the sphere products in Table~\ref{tbl:sphereproducts} are produced 
via product triangulations of boundaries of simplices \cite{lutz2003geometric}. 

In a separate experiment, we started with a triangulation of $S^1$ with $10$ vertices and with a triangulation of $S^2$ with $100$ vertices 
as the boundary complex of a random simplicial $3$-polytope, for which $100$ points on the round $2$-dimensional sphere were chosen randomly 
via the \texttt{rand\_sphere} client of the software system \texttt{polymake} \cite{polymake:2000}). The initial triangulation of $S^2 \times S^1$ 
has face-vector $f=(1000,6880,11760,5880)$. It took RSHT an average of $1108.23$ expansions, in $10^2$ runs, to reduce the triangulation (minus a facet)
to a triangulation $\partial\Delta_3 \cup K^1(21.76)$ of the wedge product $S^2 \vee S^1$. 
We repeated the same experiment, but this time applying $200,000$ preliminary random bistellar edge flips 
to the $100$-vertex triangulation of~$S^2$, before taking the sphere product. The results of this experiment are similar
to the one before (though with a slightly higher average number of expansions). This suggests that RSHT may be reliable even 
for larger complexes.

\subsection{Dimensionality reduction}

\begin{table}
\small\centering
\defaultaddspace=0.15em
\caption{RSHT run for triangulations of products of surfaces.}\label{tbl:surfaceproducts}
\vspace{-1em}
\begin{tabular*}{\linewidth}{@{\extracolsep{\fill}}l@{\hspace{1mm}}l@{\hspace{1mm}}l@{\hspace{1mm}}l@{}}
\\\toprule
 \addlinespace
 \addlinespace
  initial complex              &    initial $f$-vector                        &   resulting complex    &   resulting smallest  \\
                                       &                                                     &                                   &   $f$-vector             \\ 
\midrule
  $T \times I$                 &   $(77, 511, 854, 420)$                  &  $T$                          &  $(7, 21, 14)$    \\
  $g_2 \times I$             &   $(121, 929, 1586, 780)$             & $g_2$                       &  $(9, 32, 24)$    \\
  $g_5 \times I$             &   $(253, 2183, 3782, 1860)$          & $g_5$                       &  $(12, 60, 40)$  \\
  $g_6 \times I$             &   $(297, 2601, 4514, 2220)$          & $g_6$                       &  $(13, 69, 46)$  \\
  $g_{10} \times I$         &   $(473, 4273, 7442, 3660)$          & $g_{10}$                  &  $(18, 108, 72)$   \\
  $g_{50} \times I$         &   $(2233, 20993, 36722, 18060)$  & $g_{50}$                  &  $(51, 683, 534)$ \\
    \addlinespace
\bottomrule
\end{tabular*}
\end{table}

``Finding meaningful low-dimensional structures hidden in their high-dimensional observations'' \cite{tenenbaum2000} is a major theme in analyzing higher-dimensional data of various origins. Usually, the data is given as a finite set of points
in some Euclidean or metric space and is then often transformed to (higher-dimensional) simplicial complexes 
via taking \v{C}ech complexes or Vietoris--Rips complexes. Here, we did not start with explicit data sets, 
but instead ``hid'' a (closed) surface in a higher-dimensional product as another model to test RSHT on.

Starting with the standard $7$-vertex triangulation $T$ of the torus, we first took connected sums of $T$ to create surfaces of higher genus $g_k$, $k\geq 2$. 
Then we took the cross product $g_k\times I$  of $g_k$ with an interval (subdivided into 10 edges on 11 vertices), and reduced the resulting triangulation
of the cross product with RSHT. In every single one out of $10^2$ runs, the product $g_k\times I$ gets reduced back to a small or even vertex-minimal triangulation
of the original surface of genus $g_k$, as displayed in Table~\ref{tbl:surfaceproducts}. In a second experiment, we performed $200,000$ random edge flips
to ``randomize'' the surfaces $g_k$; then, we took  cross products with the $10$-edge interval~$I$. Again, in $10^2$ runs of RSHT, we always achieved 
the respective $f$-vectors of Table~\ref{tbl:surfaceproducts}.

In a final experiment, we started with the triangulation of the surface $g_{50}$ from before, but this time we added $100$ vertices in subdivision steps
before performing the $200,000$ random edge flips. We then took again the cross product with the interval $I$ to get a randomized triangulation of $g_{50} \times I$ 
with $f=(2728, 24278, 42212, 20760)$. We then took another cross product of this $3$-manifold with boundary with the $4$-simplex~$\Delta_{4}$.
The resulting complex is $7$-dimensional with around $34$ million faces and face vector 
$$f=(13420, 386630, 2446620, 6910210, 10432052, 8786210, 3909060, 718200).$$
In less than an hour and by using a few thousand expansions, in each out of $10^2$ runs of RSHT, 
we were able to reduce this complex back to a triangulation of the $2$-dimensional orientable surface 
of genus $50$ with fewer than $60$ vertices. In some cases we were even able to reach the same $f$-vector with $51$ vertices as in Table~\ref{figure:5Torsion}. 
Due to memory constraints that come from the computation of the Hasse diagram of the starting complex (requiring around 10 GB of RAM for this example), 
this was the largest complex that we were able to study.

\subsection{Akbulut--Kirby 4-spheres} \label{sec:AK}

As stated early on, contractibility is, in general, undecidable. However, it takes considerable effort to pose challenges to RSHT. 
A notoriously hard series of complexes is given by the triangulations of the Akbulut--Kirby $4$-dimensional spheres \cite{tsuruga2013}.
These PL-triangulated standard $4$-dimensional spheres are built in an intricate way via non-trivial presentations of the trivial group
as their fundamental group \cite{akbulut1985}. By Pachner's theorem, these examples are bistellarly equivalent to the boundary of the $5$-simplex,
and by Whitehead's theorem, the examples minus a facet are simple-homotopy equivalent to a single vertex. 
However, establishing connecting sequences of bistellar flips failed in \cite{tsuruga2013}, beyond the first easy examples of the series.
Indeed, here RSHT made no progress either, even when we set max\textunderscore step = 1,000,000 and waited for a total runtime of 60 hours.

\paragraph*{Acknowledgments.} The first author acknowledges support by NSF Grant 1855165, ``Geometric combinatorics and discrete Morse theory''. 
The third author is grateful for support by the Graduate Program ``Facets of Complexity'' (GRK 2434).

{\small
\bibliographystyle{amsalpha}
\bibliography{bibliography}
}

\end{document}